\documentclass{article}
\usepackage{spconf,graphicx}
\usepackage{lmodern}
\usepackage[latin1]{inputenc} 
\usepackage[english]{babel}
\usepackage[T1]{fontenc}

\usepackage{booktabs} 
\usepackage{multirow}
\usepackage[skip=5pt]{caption}
\usepackage{subcaption}
\usepackage{float}
\usepackage{flafter}  
\graphicspath{{figure/}} 

\usepackage{amsmath, amsthm, amssymb, amsfonts}
\numberwithin{equation}{section}
\newtheorem{theorem}{Theorem}[section]
\newtheorem{lemma}[theorem]{Lemma}

\usepackage{hyperref}

\title{A Novel Audio Representation using Space Filling Curves}
\twoauthors
  {Alessandro Mari}
	{\small EPFL, 1015, Lausanne, Switzerland\\
	\small alessandro.mari@epfl.ch}
  {Arash Salarian\sthanks{This work was supported by Logitech Europe SA.}}
	{\small Logitech Europe S.A., 1015, Lausanne, Switzerland\\
	\small arash.salarian@ieee.org}
\begin{document}
\ninept
\maketitle
\begin{abstract}

Since convolutional neural networks (CNNs) have revolutionized the image processing field, they have been widely applied in the audio context. A common approach is to convert the one-dimensional audio signal time series to two-dimensional images using a time-frequency decomposition method. Also it is common to discard the phase information. In this paper, we propose to map one-dimensional audio waveforms to two-dimensional images using space filling curves (SFCs). These mappings do not compress the input signal, while preserving its local structure. Moreover, the mappings benefit from progress made in deep learning and the large collection of existing computer vision networks.
We test eight SFCs on two keyword spotting problems. We show that the Z curve yields the best results due to its shift equivariance under convolution operations. Additionally, the Z curve produces comparable results to the widely used mel frequency cepstral coefficients across multiple CNNs.
\end{abstract}
\section{Introduction}
\label{sec:intro}

The first step when building an audio model is to choose the input space. 
Usually the data is pre-processed to obtain an intermediate low-level representation, which is subsequently used as model input. Most of these methods are frequency-based. For example, Tanweer et al. \cite{mfcc_classification} used mel frequency cepstral coefficients (MFCCs) for speech recognition  \cite{mfcc}. 
These coefficients are computed over a sliding window and stacked together to obtain a two-dimensional time-frequency image. During the computation, only the magnitude of the complex numbers is kept and the phase is discarded. This approach has shown some limitations in speech enhancement \cite{phase} and performing the Fourier transform adds extra computational costs.

Since the success of deep neural networks in image classification \cite{alexnet} and the arrival of dedicated hardware to train large models in parallel, deep neural networks have been widely applied to audio signals outperforming previous approaches \cite{hinton}. A large variety of image networks have been proposed, such as convolutional neural networks (CNN, \cite{lenet}).
These can be used in combination with time-frequency images, although the two axes are semantically different from the horizontal and vertical axes of an image.


Separating the construction of an appropriate audio representation from the design of the model architecture might not be optimal for the task at hand. Hence some authors (e.g. \cite{raw_audio,raw_audio_2}) used raw audio waveform representations as inputs of one-dimensional CNNs.
Although the additional cost of the pre-processing step is suppressed, connecting signals that are far apart, typically requires deeper architectures to increase the receptive field (RF).

In this study, we investigate a novel approach as an alternative to frequency domain based inputs and to raw audio waveforms. Space filling curves (SFC) enable us to map audio waveforms to two-dimensional images. In our approach, the input signal is not compressed and no information is lost since the SFCs act as bijective maps. By converting the audio samples to two-dimensional images, it is possible to leverage advanced deep neural networks from Computer Vision (CV).
SFCs also reduce the distance between indices compared to one-dimensional indexing schemes, which guarantees the same receptive field with fewer layers.
 Finally, in a potential hardware implementation, the mapping is constant and does not need any runtime computations.

SFCs have been used together with CNNs in the past. For example, Yin et al. \cite{dna} have used the Hilbert curve \cite{hilbert} to combine consecutive representation of $k$-mers for the prediction of the chromatin state of a DNA sequence. Tsinganos et al. \cite{sEMG} have used the Hilbert curve to associate surface electromyography (sEMG) signals with hand gestures. In their experiments, the input image was obtained by mapping the time series into a two-dimensional image for each sEMG channel, which is similar to our technique. SFCs were also used for malware classification and detection \cite{malware1,malware2}. In their work the authors have mapped the code of the program, which can be viewed as a sequence of bytes, to pixels of an image. To the best of our knowledge, SFC mappings have not been used in the context of audio representation.

\section{Space Filling Curves}

In this paper, a SFC is a bijective mapping $C_k: [N_k^2] \rightarrow  [N_k] \times  [N_k]$, where $N_k=b^k$, $[n]=\{1,\dots,n\}$ and the SFC image representation of an audio sample $s=\{s_i\}_{i=0}^L$ is obtained by $I_k(s)(i,j)  =s_{C^{-1}_k(i,j)}$. A SFC can be interpreted as an ordering of the pixels of an image, or as an indexing scheme for time series. Note that this definition differs from the usual one given by Peano \cite{peano}, which refers to the limiting process ($k\rightarrow \infty$) after a proper normalization of the input and the output. We can distinguish two families of SFCs: recursive space filling curves (RSFC) and non-recursive space filling curves (NRSFC).

\subsection{Recursive space filling curves}
A RSFC is built recursively, where $C_{k+1}$ is obtained by subdividing the cell $(i,j)$ of the curve $C_k$ into $b^2$ cells and by modifying the curve $C_{k}$ according to a set of rules. 
A large number of RSFCs exist. In this paper, we focus on the following RSFCs: the Hilbert curve \cite{hilbert}, the Z curve (also known as the Z-order \cite{z}), the Gray curve \cite{gray_curve}, the H curve \cite{H_curve}, and a curve, that we will call OptR, proposed by Asano et al. \cite{asano_space-filling_1997}. A representation of these curves with $k=3$ can be found in Figure \ref{fig:rsfc}. Consecutive points $C_k(i),C_k(i+1)$ are connected by a line. Dotted lines represent jumps.

\begin{figure}[ht]
    \centering
    \includegraphics[width=0.88\linewidth]{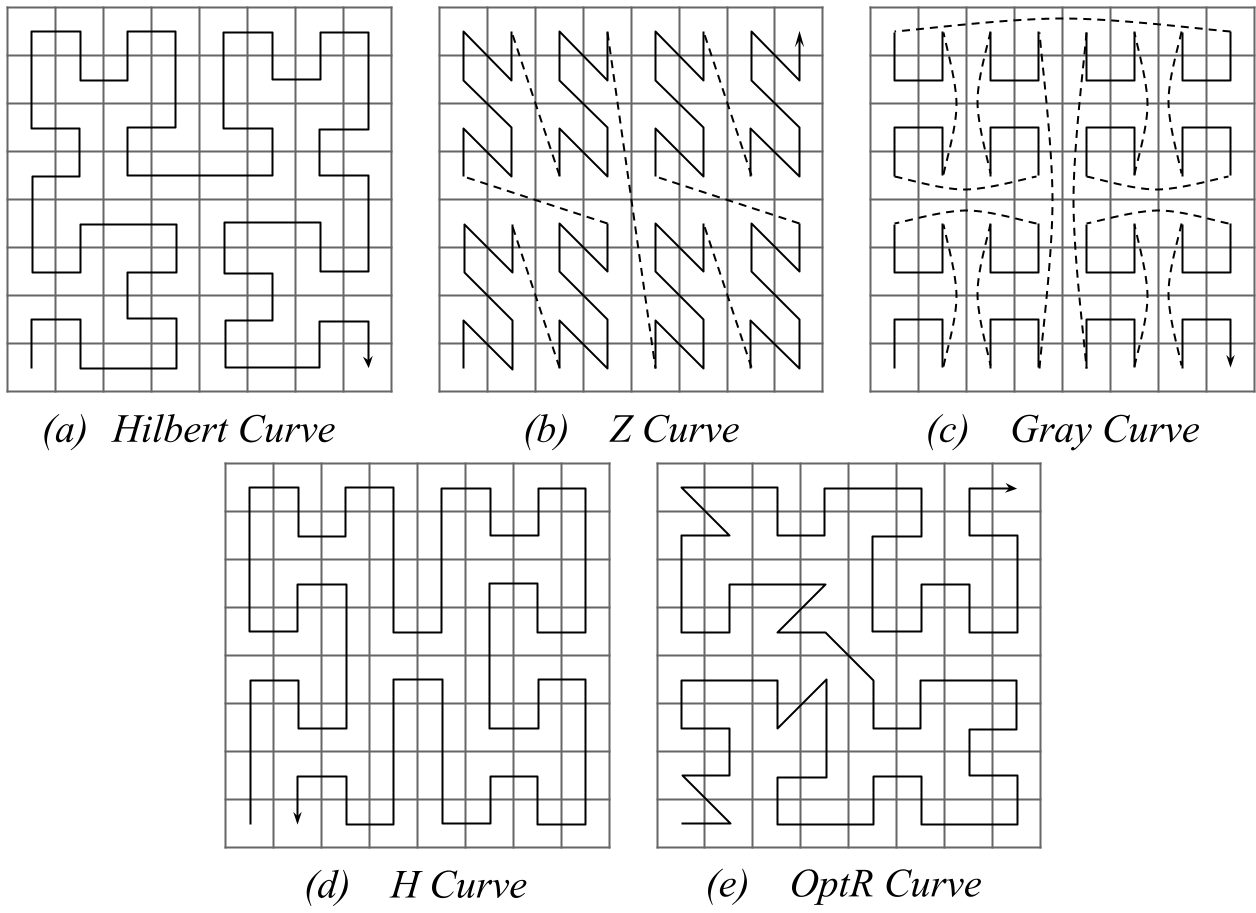}
    \caption{Recursive space filling curves of order $k=3$.}
    \label{fig:rsfc}
\end{figure}

The Hilbert curve shows good performance in locality preservation, which is desirable in the context of audio samples: data points close in time should be close in the image representation in order to exploit the local nature of the convolution layers. Depending on the locality metric, the Hilbert curve outperforms the Z curve and the Gray curve \cite{rsfc_database}, but is also inferior to the Z curve \cite{locality}.
Nonetheless, the Z curve and the Gray curve include jumps between consecutive indices, which don't preserve locality. We included both of them to study the impact of this behavior on models.

Niedermeier et al. \cite{H_curve} have shown that the H curve almost reaches the optimal lower bound among closed cyclic curves (i.e. $\|C_k(i)-C_k(i+1)\|_\infty=1$, $\forall i$). In particular, they have shown that $\|C_k(i)-C_k(j)\|_p$, $p=1,2,\infty$ roughly behaves like $\sqrt{a|i-j|+b}$ in the worst case. This bound is smaller than $|i-j|$, which is the lower bound in the one-dimensional case  achieved by the identity mapping. It shows the benefit from using two-dimensional inputs over one-dimensional inputs, as the distance between indices is reduced. 
Note that the OptR curve is particularly complex, which might lead to poor performance, since regular patterns of the curve are harder to identify.

\subsection{Non-recursive space filling curve}

This category groups all curves that cannot be built recursively. 
We investigate three additional curves: the Scan curve, the Sweep curve, and the Diagonal curve \cite{gis}.
Figure \ref{fig:non_rec_curve} shows their representation when $k=3$. The Sweep curve is a two-dimensional ordering scheme, where the sequence is cut into equally long intervals and stacked together to obtain an image. 
The Scan curve is obtained by reversing one interval out of two, which removes all jumps. Finally, the Diagonal curve is a 45 degrees rotated version of the Scan curve restricted to the $N_k\times N_k$ grid interior. 
These curves do not have particular locality preservation properties apart from being continuous mappings, i.e. $\|C_k(i)-C_k(i+1)\|_\infty=1$ (except for Sweep). Moreover, $ \|C_k(i)-C_k(j)\|_\infty$ scales like $|i-j|$ in the worst case, which is much larger than $\sqrt{|i-j|}$.

\begin{figure}[ht]
    \centering
    \includegraphics[width=0.9\linewidth]{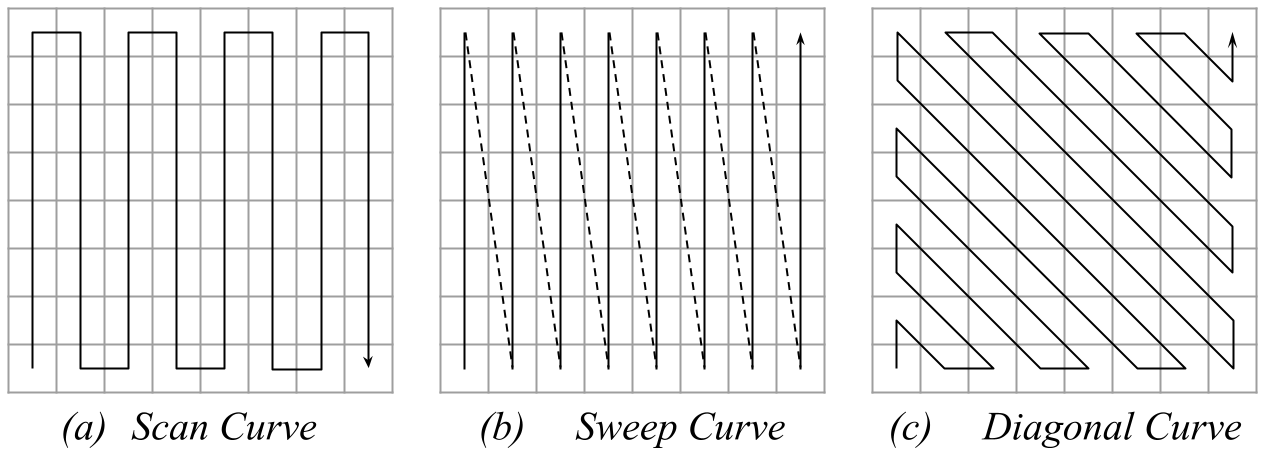}
    \caption{Non-recursive space filling curves of order $k=3$.}
    \label{fig:non_rec_curve}
\end{figure}

\subsection{Pre-processing steps}\label{sec:preprocessing}

We distinguish two pre-processing steps: functions $f\in \mathcal{P}$ that are applied on the raw audio signal and transformations $g\in \mathcal{T}$ that are performed on images. 
Assuming that the SFC mapping could be implemented in hardware,
transformations in $\mathcal{T}$ should be preferred over functions in $\mathcal{P}$.



We might first \texttt{center} audio recordings in the middle of the frame. The procedure consists of first computing the weighted average energy in separate windows of size $w$ using a Gaussian kernel with standard deviation $\sigma$, and aligning the signal that is above the threshold $th$ with the middle of the output sequence.
Note that $\texttt{center}\in \mathcal{P}$.


We might also use two additional data augmentations: \texttt{mixup} \cite{mixup} and \texttt{shift}.
\texttt{Mixup} encourages the model to behave linearly in-between two training examples by taking random convex combinations. This operation is parametrized by $\lambda$, which follows a Beta distribution $\lambda \sim \mathcal{B}(\alpha,\alpha)$, where $\alpha>0$ is some hyperparameter. In the audio setting, this method can be applied either to the raw audio waveform or to the image representation. Since the SFC mapping is a linear operator, both approaches are equivalent. We choose to apply it to the image (i.e. $\texttt{mixup} \in \mathcal{T}$).
\texttt{Shift} consists of applying random shifts $s_u$ from a uniform distribution $s_u\sim \mathcal{U}(-L/4,L/4)$, where $L$ is the input length,  to reduce the location dependency.
This operation requires that the input is already centered in order to lose less information.

\section{Experiments and Results}

\subsection{Experimental setup}

We used MFCC, one of the most common audio representations, as a baseline in this study\footnote{Code available at \url{https://github.com/amari97/sfc-audio}}. We focused on keyword spotting using two different datasets: Google Speech Commands V2 \cite{speechcommands} and a dataset based on the 1000 most frequent words of LibriSpeech's ``clean'' split \cite{librispeech}, for which Beckmann et al.  \cite{beckmann} have provided aligned labels\footnote{Available at \url{https://github.com/bepierre/SpeechVGG}}. The Speech Commands dataset contains 105829 words classified into 35 classes, where 84843, 9981, 11005 of them were respectively used for training, validation, and test. Among the 1000 most frequent words in LibriSpeech, we have selected 31 words that were appearing between 4000 and 8000 times in order to have a balanced dataset and spoken words of similar duration. 
We ended up with 164674, 2519, 2400 training/validation/test examples. 
All audio recordings were sampled at 16 kHz and had a duration of one second. 

To properly compare both approaches, chosen models should process inputs of different sizes without changing their structure. This was achieved by introducing an average pooling layer with an area equal to the input size.
Several CV networks using this layer were tested: the small MobileNetV3 network (\texttt{Mo}) \cite{mobilenetv3}, ShuffleNetV2 (\texttt{Shuffle}) \cite{shufflenetv2}, SqueezeNet (\texttt{Squeeze}) \cite{squeezenet}, MixNet (\texttt{Mix}) \cite{mixnet} and EfficientNet B0 (\texttt{Eff}) \cite{efficientnet}. Additionally, the representations were evaluated using \texttt{Res8} \cite{res8}, a small network that was specifically designed for keyword spotting tasks with MFCCs. Few modification were needed since the models had different input requirements. Since \texttt{Mo} accepts three channel images (RGB), we added a pointwise convolution to expand the input, followed by a $\texttt{ReLU}$ activation layer and $\texttt{BatchNorm}$ normalization layer \cite{batchnorm}. Our implementation of \texttt{Squeeze} added \texttt{BatchNorms} after the squeezing layer and the classifier was replaced by a fully connected layer.\footnote{We use the version 1.1 as described in the official code \url{https://github.com/forresti/SqueezeNet}.} Finally, the pooling area of \texttt{Res8} was extended to $4\times 4$ to increase the RF, and the dilation of the convolutions was set to 1,1,1,2,2,2, following the increasing dilation strategy of the Res15 network \cite{res8}. As a result, the RF of \texttt{Res8} was changed from 54 to 78.

Unless specified otherwise, all models were trained with stochastic gradient descent (SGD) with learning rate 0.5 and batch size 256 on 2 GPUs with distributed data parallel. We used an early stopping rule with 50 waiting epochs and we set to 300 the maximal number of epochs. We set the \texttt{mixup} hyperparameter $\alpha=0.2$, as in \cite{mixup}. We set $th=0.0001$, such that only signals that were 1.6 higher than a flat audio signal exceeded it. We also set the Gaussian kernel parameters to $w=100$ and $\sigma=25$, such that boundary values had a weight close to zero.
We used 40 MFCCs that were computed on frames of 0.025 seconds with overlaps of 0.015 seconds \cite{res8,accurate_mfcc}. 
Finally, we used curves of order $k=7$ to represent one second of audio, since $16000<16384=2^7\times 2^7$.

\subsection{Results}

\textbf{Curve comparison} As shown in Table \ref{tab:acc_curves_mobile}, the data augmentations improved the accuracy of the model, even if this effect was slightly less evident in the case of LibriSpeech. The benefits of these augmentations were larger for SFCs: the generated images were very different, improving the selection of invariant features. The Z curve yielded the best results in all situations and reached the baseline performance with \texttt{Data Aug.} 
Non-recursive curves performed equally well, but were inferior to Z and H. Finally, OptR was the worst one.

\begin{table}[ht]
    \centering
\resizebox{0.9\columnwidth}{!}{
    \begin{tabular}{clcccc}
    \toprule
       & Dataset &\multicolumn{2}{c}{Speech Commands}  & \multicolumn{2}{c}{LibriSpeech}   \\
       \cmidrule(lr){3-4}\cmidrule(lr){5-6}
       && SGD* & Data Aug.& SGD & Data Aug.\\
       \midrule
       &MFCC &89.3 &\textbf{93.0}  &97.9 & \textbf{98.2} \\
        \midrule
        \multirow{5}{*}{\rotatebox[origin=c]{90}{Recursive}}&Hilbert & 80.1& 89.0 &94.9 &94.1 \\
       &Z & 86.2 &\textbf{92.8} &97.4 &\textbf{98.3} \\
       &Gray &82.5 &90.4 &97.0 &96.2 \\
       &H &82.8 &91.6 &96.0 &97.8 \\
       &OptR &78.9 &88.1 &93.8 & 93.0 \\
       \midrule
       \multirow{3}{*}{\rotatebox[origin=c]{90}{Non-rec.}}&Sweep & 83.2 &90.3 &94.3 &95.5 \\
       &Scan & 83.8& 90.3&94.7 &96.1 \\
       &Diagonal & 83.4& 90.0 &95.9 &95.6 \\
        \bottomrule
    \end{tabular}}
    \caption{Test accuracy of each SFC on \texttt{Mo}. \texttt{SGD} is the baseline that only includes the SFC/MFCC computation. \texttt{Data Aug.} includes \texttt{mixup}, \texttt{shift} and \texttt{center}. *trained with an early stopping rule with 100 waiting epochs.}
    \label{tab:acc_curves_mobile}

\end{table}


\textbf{Ablation study} As shown in Table \ref{tab:acc_transformation_mobile},
centering the audio clip inside the one second frame improved the accuracy for almost all curves. In particular, the improvement was larger for recursive curves. The \texttt{Shift} step provided a large gain by forcing the model to select features that are shift invariant in time. \texttt{Mixup} still improved the previous accuracy.

\begin{table}[ht]
    \centering
    \setlength{\tabcolsep}{1.5pt}
    
    \resizebox{\columnwidth}{!}{
    \begin{tabular}{p{38mm}cccccccc}
    \toprule
       &  \multicolumn{5}{c}{Recursive}  & \multicolumn{3}{c}{Non-recursive}   \\
       \cmidrule(lr){2-6}\cmidrule(lr){7-9}
       & Hilbert & Z & Gray & H & OptR & Sweep & Scan & Diagonal \\
       \midrule
       SGD* & 80.1&86.2 &82.5 &82.8 &78.9 &83.2 &83.8 &83.4 \\
        
        + Center*& 83.1& 87.3& 84.4& 84.4&81.8 &83.5 &84.6&83.3 \\
       
        
        + Center + Shift & 84.5& 90.1& 87.9& 86.0& 85.8& 86.8& 87.1& 87.2\\
        + Center + Shift + Mixup &89.0 &92.8 &90.4 &91.6 &88.1 &90.3 &90.3 &90.0 \\
        \bottomrule
    \end{tabular}}
    \caption{Test accuracy on \texttt{Mo}. \texttt{SGD} is the baseline that only includes the SFC/MFCC computation. *trained with an early stopping rule with 100 waiting epochs.}
    \label{tab:acc_transformation_mobile}
\end{table}

\textbf{Model comparison} Table \ref{tab:model_comp} gives the results of the comparison between the Z curve and the MFCC approach. Except for \texttt{Res8}, there was no difference with the baseline. The best results were unexpectedly obtained with the largest network (\texttt{Eff}), but \texttt{Res8} achieved impressive results in terms of efficiency and accuracy with MFCCs.

\begin{table}[ht]
    \setlength{\tabcolsep}{1.5pt}
    \centering
    \resizebox{\columnwidth}{!}{
    \begin{tabular}{llcccccc}
    \toprule
        & & Res8 & Mo & Shuffle & Squeeze & Mix & Eff  \\
        \midrule
        Parameters& & 111K&1'554K&1'289K&740K&2'653K&4'052K\\
         \midrule
        \multirow{2}{*}{Speech Commands}&MFCC &94.0 &93.0 &92.9 &93.9 &93.9 &95.1 \\
                                        &Z & 85.3&92.8 &92.0 &91.2* &94.1 & 94.9\\
         \midrule
          \multirow{2}{*}{LibriSpeech}&MFCC &98.6 & 98.2 & 98.1 & 98.1 &98.8 & 99.0\\
                                        &Z &95.9 & 98.3& 97.8 & 97.9* & 98.5& 99.2 \\
        \bottomrule
    \end{tabular}}
    \caption{Test accuracy with \texttt{Data Aug.} *trained with an early stopping rule with 100 waiting epochs.}
    \label{tab:model_comp}
\end{table}

\textbf{Receptive field influence} 
Figure \ref{fig:receptive_field} shows the average output probability $\bar{p}(s)$ of the true class as a function of time shifts $s$. 
The figure revealed that $\bar{p}(s)$ was lower when the audio signal was centered ($s=0$), and when the RF was much smaller than the size of the input image ($78<128$). This effect was alleviated when increasing the RF to 125 by setting the dilation of the \texttt{Res8} filters to $1,1,2,2,4,4$. In this case, the model accuracy reached 86.5.

\begin{figure}[ht]
    \centering
    \includegraphics[width=0.65\linewidth]{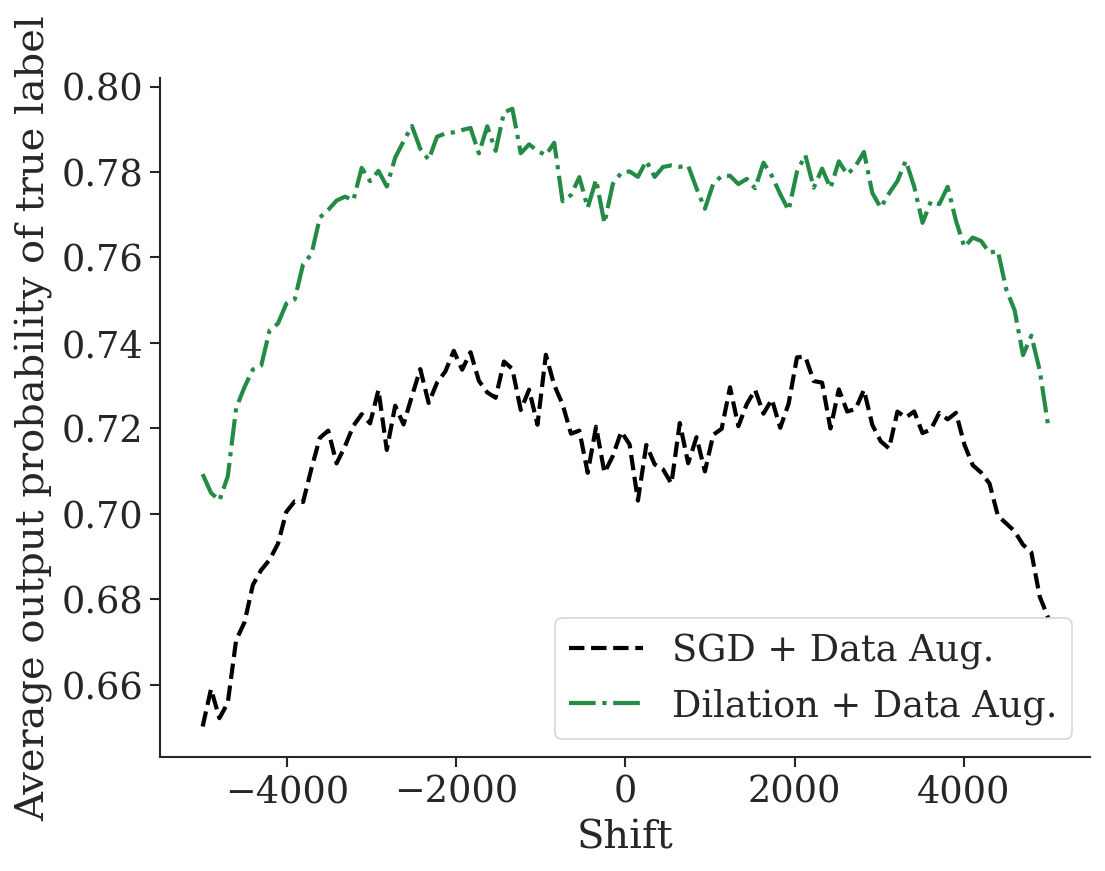}
    \caption{Average output probabilities of the true label when shifting the input sequence relatively to the center of the frame. Results obtained with \texttt{Res8} on Speech Commands.}
    \label{fig:receptive_field}
\end{figure}

\textbf{Number of parameters} Figure \ref{fig:width_mult} shows the model accuracy after changing the number of parameters by shrinking or expanding the width of the network (i.e. the number of channels of each convolution layer) by a factor $\texttt{width\_mult}$. The large gap with the baseline persisted on \texttt{Res8}. Moreover, the difference remained similar when reducing the number of parameters. Due to its prohibitively large computational cost, the standard deviations were only computed for \texttt{Mo} using cross-validation, but they showed that the gap with the baseline was not significant.

\begin{figure}[ht]
   \begin{minipage}{.68\linewidth}
    
         \includegraphics[width=1.0\linewidth]{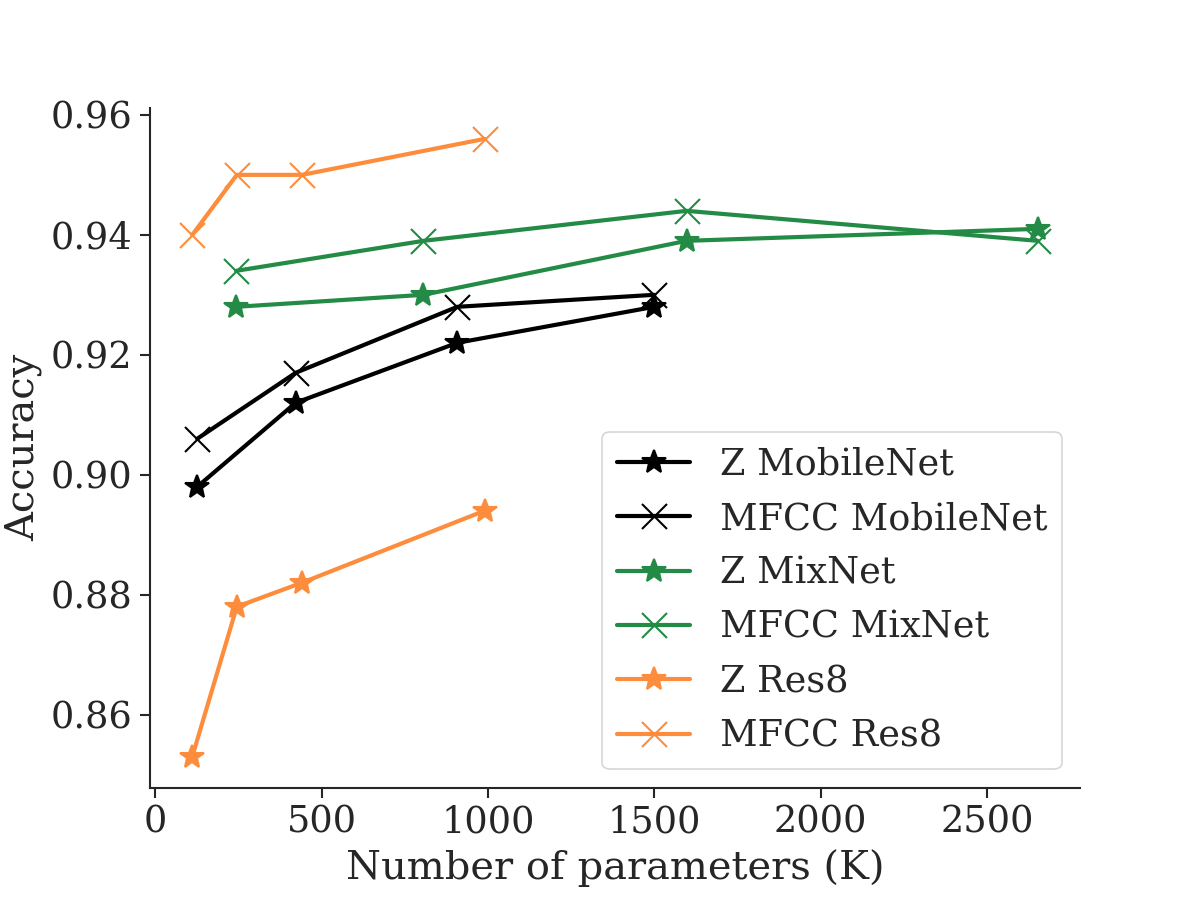}
    \end{minipage}
    \begin{minipage}{.30\linewidth}
    \centering
    \setlength{\tabcolsep}{1pt}
    \resizebox{\columnwidth}{!}{
    \begin{tabular}{p{8mm}cc}
    \toprule
       \texttt{width mult.}&Z&MFCC\\
       \midrule
       0.25&0.011&0.012\\
       0.5&0.011&0.009\\
       0.75&0.008&0.004\\
       1.0&0.013&0.006\\
        \bottomrule
    \end{tabular}}
    \end{minipage}
    \caption{\textit{Left:} Test accuracy of \texttt{Mo}, \texttt{Mix} with $\texttt{width\_mult}\in \{0.25,0.5,0.75,1\}$ and \texttt{Res8} with  $\texttt{width\_mult}\in\{1,1.5,2,3\}$. Models are trained with \texttt{Data Aug.} on Speech Commands. \textit{Right:} Standard deviation of \texttt{Mo} accuracy using 10-fold cross-validation.}
    \label{fig:width_mult}
\end{figure}

\textbf{Curve Comparison on Res8} Table \ref{tab:acc_res8} shows the performance of each curve with \texttt{Res8}. Non-recursive curves yielded the best results, while the $Z$ curve outperformed the other recursive curves. The gap with the baseline was significant.

\begin{table}[ht]
    \centering
    \setlength{\tabcolsep}{1pt}
    
    \resizebox{\columnwidth}{!}{
    \begin{tabular}{p{27mm}ccccccccc}
    \toprule
       &&  \multicolumn{5}{c}{Recursive}  & \multicolumn{3}{c}{Non-recursive}   \\
       \cmidrule(lr){3-7}\cmidrule(lr){8-10}
       &MFCC& Hilbert & Z & Gray & H & OptR & Sweep & Scan & Diagonal \\
       \midrule
       Speech Commands& 94.0 &80.0 &85.3 &76.0 &80.2 &77.7 &88.3 &89.7 &89.8 \\
       LibriSpeech&98.6 &92.5 &95.9 &92.8 &93.0 &90.6 &96.5 &97.2 &96.9 \\
        \bottomrule
    \end{tabular}}
    \caption{Test Accuracy on \texttt{Res8}. Models were trained with \texttt{Data Aug.}}
    \label{tab:acc_res8}
\end{table}

\subsection{Discussion}

Despite having jumps between consecutive indices, which could harm the model performance (see Figure \ref{fig:receptive_field}), the Z curve showed the best performance among other SFCs. Our intuition is that the Z curve guarantees a rather simple relation between the features of the hidden layers, when randomly shifting the input in time. In particular, the Z curve exhibits a shift equivariance under convolution operations due to its regular patterns (Z shape) that are always oriented in the same direction (Figure \ref{fig:rsfc}). More precisely, 

\begin{lemma}
Let $W_l$ be a $2^l\times 2^l$ discrete convolution with stride $2^l$ in $x,y$ direction on a $2^k\times 2^k$ image $I$, $k>l$, namely
$$O_l(I)(i,j)=\sum_{m,n\in [2^l]} I(i2^l+m,j2^l+n)\times W_l(m,n).$$
Let also $s$ be a sequence of numbers of length $2^{2k}$ and let 
$I_k(s)(i,j)=s_{C_k^{-1}(i,j)}$ be the corresponding image using the Z curve $C_k$. Then, if the sequence $s$ is circularly shifted by $r=d2^{2l}$, for some $d\in \mathbb{N}$ (i.e. $s^{r}(i)=s(i+r \mod 2^{2k})$), then $I_{k-l}^{-1}\circ O_l\circ I_k(s^r)$ is equal to $I_{k-l}^{-1}\circ O_l\circ I_k(s)$ up to a shift of $d$ units.
\end{lemma}
\begin{proof}
Note that the $Z$ mapping of some index $z$, expressed in binary $z\equiv z_{2k}\dots z_0$, $z_i\in \{0,1\}$, is obtained by interleaving the bits of its binary expression, namely $x=z_{2k-1}\dots z_1,\  y=z_{2k}\dots z_0$. Let $i\equiv a_{m} \dots a_0$, $j\equiv b_{m} \dots b_0$, we have $I_k(s)(i,j)=s_{a_m b_m\dots a_0b_0}$. Shifting the output sequence by $d2^{2l}\equiv d_{2k}\dots d_{2l} 0 \dots 0$ doesn't modify the digits $a_{l-1} \dots a_0$, $b_{l-1} \dots b_0$. Therefore $I_k(s^r)(i,j)=s_{a_m b_m\dots a_0b_0}^r=s_{\tilde{a}_m \tilde{b}_m \dots a_{l-1} b_{l-1}\dots a_0b_0}$. Similarly, using a stride of $2^l$ only modifies digits that are after the $2l$-th index. Let's consider the case $i=0,j=0$. Then 
$O_l\circ I_k(s^r)(0,0)=O_l\circ I_k(s)(d_{2k-1}\dots d_{2l+1},d_{2k} \dots d_{2l})$. Unfolding the image finally yields $I_{k-l}^{-1}\circ O_l \circ I_k(s) (d_{2k}\dots d_{2l})=I_{k-l}^{-1}( O_l \circ I_k(s)(d_{2k-1}\dots d_{2l+1},d_{2k} \dots d_{2l}))=I_{k-l}^{-1}(O_l\circ I_k(s^r)(0,0))=I_{k-l}^{-1}\circ O_l \circ I_k(s)(0)$, which proves the lemma for $i=0, j=0$. For the general case the result follows because we have $O_l\circ I_k(s^r)(i,j)=O_l\circ I_k(s)(\tilde{a}_m\dots\tilde{a}_l,\tilde{b}_m \dots\tilde{b}_l)$.
\end{proof}
On the other hand, curves like the Hilbert curve (Figure \ref{fig:rsfc}) include some rotation of its elementary block, which destroys the equivariance property. This intuition is further assessed by the fact that the Z curve performed well without imposing some shift invariance with data augmentations (Table \ref{tab:acc_curves_mobile}), which shows that the model has already built a coherent feature extraction.


Table \ref{tab:model_comp} shows that the SFC approach was competitive with the baseline when trained with \texttt{Data Aug.}, except for the \texttt{Res8} network, for which the large jump in the middle of the Z mapping had an influence on the accuracy (Figure \ref{fig:receptive_field}). Increasing the RF of the network improved the accuracy from 85.3 to 86.5, but remained substantially smaller than the baseline ($86.5\ll 94.0$). The number of parameters was not sufficient to justify the discrepancy observed in Table \ref{tab:model_comp}, since the gap persisted on Figure \ref{fig:width_mult}.
It might be related to the network architecture, which was specifically designed for MFCC inputs and which contained a small number of layers.
Indeed \texttt{Res8} worked best with non-recursive curves (Table \ref{tab:acc_res8}), which have an image structure similar to the MFCC representation: the image can be decomposed into two perpendicular axes --- a time axis and a feature axis --- which is not possible for the recursive curves (see Figure \ref{fig:rsfc}).

\section{Conclusion}
We have proposed an alternative audio representation to frequency-based images and raw audio waveforms using space filling curves. We have shown that it achieves comparable performance to the widely used MFCC representation when combined with deep CNNs in the context of keyword spotting tasks. In particular, the Z curve yields the best results, which is probably due to its shift equivariance under convolution operations.
Our study suggest that to leverage DNNs time-frequency decomposition should not be considered as a central dogma and simpler one-dimensional to two-dimensional mappings such SFCs might perform just as well. Future work could aim at checking the robustness of the image representation under noisy inputs and at generalizing the method to variable input lengths. 


%
%
%


\bibliographystyle{IEEEbib}
\bibliography{refs}

\begin{thebibliography}{10}

\bibitem{mfcc_classification}
S.~Tanweer, A.~Mobin, and A.~Alam,
\newblock ``{Analysis of Combined Use of NN and MFCC for Speech Recognition},''
\newblock {\em International Journal of Computer and Information Engineering},
  vol. 8, pp. 1736--1739, 2015.

\bibitem{mfcc}
{European Telecommunications Standards Institute},
\newblock ``{Speech} {Processing}, {Transmission} and {Quality} {Aspects}
  ({STQ}); {Distributed} speech recognition; {Front}-end feature extraction
  algorithm; {Compression} algorithms,''
\newblock Tech. {R}ep. {ES}-201-108, v1.1.3, ETSI, 2003.

\bibitem{phase}
K.~Paliwal, K.~W\'{o}jcicki, and B.~Shannon,
\newblock ``The importance of phase in speech enhancement,''
\newblock {\em Speech Communication}, vol. 53, no. 4, pp. 465--494, 2011.

\bibitem{alexnet}
A.~Krizhevsky, I.~Sutskever, and G.~E. Hinton,
\newblock ``{ImageNet Classification with Deep Convolutional Neural
  Networks},''
\newblock in {\em Proc. NIPS}, 2012, vol.~1, pp. 1097--1105.

\bibitem{hinton}
G.~Hinton, L.~Deng, D.~Yu, G.~E Dahl, A.~Mohamed, N.~Jaitly, A.~Senior,
  V.~Vanhoucke, P.~Nguyen, T.~N. Sainath, et~al.,
\newblock ``Deep neural networks for acoustic modeling in speech recognition:
  The shared views of four research groups,''
\newblock {\em IEEE Signal processing magazine}, vol. 29, no. 6, pp. 82--97,
  2012.

\bibitem{lenet}
Y.~LeCun, B.~Boser, J.~S. Denker, D.~Henderson, R.~E. Howard, W.~Hubbard, and
  L.~D. Jackel,
\newblock ``{Backpropagation Applied to Handwritten Zip Code Recognition},''
\newblock {\em Neural Computation}, vol. 1, no. 4, pp. 541--551, 1989.

\bibitem{raw_audio}
J.~Lee, T.~Kim, J.~Park, and J.~Nam,
\newblock ``Raw {Waveform}-based {Audio} {Classification} {Using}
  {Sample}-level {CNN} {Architectures},''
\newblock {\em arXiv preprint arXiv:1712.00866}, 2017.

\bibitem{raw_audio_2}
P.~Ghahremani, V.~Manohar, D.~Povey, and S.~Khudanpur,
\newblock ``{Acoustic Modelling from the Signal Domain Using CNNs},''
\newblock in {\em Proc. Interspeech}, 2016, pp. 3434--3438.

\bibitem{dna}
B.~Yin, M.~Balvert, D.~Zambrano, A.~Schoenhuth, and S.~Bohte,
\newblock ``An image representation based convolutional network for {DNA}
  classification,''
\newblock in {\em Proc. ICLR}, 2018.

\bibitem{hilbert}
D.~Hilbert,
\newblock {\em {\"U}ber die stetige Abbildung einer Linie auf ein
  Fl{\"a}chenst{\"u}ck}, pp. 1--2,
\newblock Springer, Berlin, Heidelberg, 1935.

\bibitem{sEMG}
P.~Tsinganos, B.~Cornelis, J.~Cornelis, B.~Jansen, and A.~Skodras,
\newblock ``A hilbert curve based representation of semg signals for gesture
  recognition,''
\newblock in {\em Proc. IWSSIP}, 2019, pp. 201--206.

\bibitem{malware1}
Z.~Ren, G.~Chen, and W.~Lu,
\newblock ``Malware visualization methods based on deep convolution neural
  networks,''
\newblock {\em Multimedia Tools and Applications}, pp. 1--19, 2019.

\bibitem{malware2}
S.~O'Shaughnessy,
\newblock ``{Image-based Malware Classification: A Space Filling Curve
  Approach},''
\newblock in {\em {IEEE Symposium on Visualization for Cyber Security}}, 2019,
  pp. 1--10.

\bibitem{peano}
G.~Peano,
\newblock ``Sur une courbe, qui remplit toute une aire plane,''
\newblock {\em Mathematische Annalen}, vol. 36, no. 1, pp. 157--160, 1890.

\bibitem{z}
G.~M. Morton,
\newblock ``A computer oriented geodetic data base and a new technique in file
  sequencing,''
\newblock Tech. {R}ep., IBM Ltd., Ottawa, 1966.

\bibitem{gray_curve}
C.~Faloutsos,
\newblock ``{Multiattribute Hashing Using Gray Codes},''
\newblock in {\em Proc. SIGMOD}, 1986, pp. 227--238.

\bibitem{H_curve}
R.~Niedermeier, K.~Reinhardt, and P.~Sanders,
\newblock ``Towards optimal locality in mesh-indexings,''
\newblock {\em Discrete Applied Mathematics}, vol. 117, no. 1, pp. 211--237,
  2002.

\bibitem{asano_space-filling_1997}
T.~Asano, D.~Ranjan, T.~Roos, E.~Welzl, and P.~Widmayer,
\newblock ``Space-filling curves and their use in the design of geometric data
  structures,''
\newblock {\em Theoretical Computer Science}, vol. 181, no. 1, pp. 3--15, 1997.

\bibitem{rsfc_database}
B.~Moon, H.~V. Jagadish, C.~Faloutsos, and J.~H. Saltz,
\newblock ``{Analysis of the clustering properties of the Hilbert space-filling
  curve},''
\newblock {\em IEEE Transactions on Knowledge and Data Engineering}, vol. 13,
  no. 1, pp. 124--141, 2001.

\bibitem{locality}
H.~K. Dai and H.~C. Su,
\newblock ``On the {Locality} {Properties} of {Space}-{Filling} {Curves},''
\newblock in {\em International Symposium on {Algorithms} and {Computation}}.
  2003, pp. 385--394, Springer, Berlin, Heidelberg.

\bibitem{gis}
M.~F. Mokbel and W.~G. Aref,
\newblock ``Space-filling curves,''
\newblock in {\em Encyclopedia of GIS}, S.~Shekhar and H.~Xiong, Eds., pp.
  1068--1072. Springer US, Boston, MA, 2008.

\bibitem{mixup}
H.~Zhang, M.~Cisse, Y.~N. Dauphin, and D.~Lopez-Paz,
\newblock ``mixup: {Beyond} {Empirical} {Risk} {Minimization},''
\newblock in {\em Proc. ICLR}, 2018.

\bibitem{speechcommands}
P.~Warden,
\newblock ``Speech commands: A dataset for limited-vocabulary speech
  recognition,''
\newblock {\em arXiv preprint arXiv:1804.03209}, 2018.

\bibitem{librispeech}
V.~Panayotov, G.~Chen, D.~Povey, and S.~Khudanpur,
\newblock ``{Librispeech: An ASR corpus based on public domain audio books},''
\newblock in {\em Proc. ICASSP}, 2015, pp. 5206--5210.

\bibitem{beckmann}
P.~Beckmann, M.~Kegler, and M.~Cernak,
\newblock ``Word-level embeddings for cross-task transfer learning in speech
  processing,''
\newblock {\em EUSIPCO (Submitted)}, 2021.

\bibitem{mobilenetv3}
A.~Howard, M.~Sandler, G.~Chu, L.~Chen, B.~Chen, M.~Tan, W.~Wang, Y.~Zhu,
  R.~Pang, V.~Vasudevan, Q.~V. Le, and H.~Adam,
\newblock ``{Searching for MobileNetV3},''
\newblock in {\em Proc. ICCV}, 2019, pp. 1314--1324.

\bibitem{shufflenetv2}
N.~Ma, X.~Zhang, H.~Zheng, and J.~Sun,
\newblock ``{ShuffleNet V2: Practical Guidelines for Efficient CNN Architecture
  Design},''
\newblock in {\em Proc. ECCV}, 2018, pp. 122--138.

\bibitem{squeezenet}
F.~Iandola, S.~Han, M.~Moskewicz, K.~Ashraf, W.~Dally, and K.~Keutzer,
\newblock ``{SqueezeNet: AlexNet-level accuracy with 50x fewer parameters and
  <0.5MB model size},''
\newblock {\em arXiv preprint arXiv:1602.07360}, 2016.

\bibitem{mixnet}
M.~Tan and Q.~V. Le,
\newblock ``{MixConv: Mixed Depthwise Convolutional Kernels},''
\newblock in {\em BMVC}, 2019, p.~74.

\bibitem{efficientnet}
M.~Tan and Q:~V. Le,
\newblock ``{E}fficient{N}et: Rethinking model scaling for convolutional neural
  networks,''
\newblock in {\em Proc. ICML}, 2019, vol.~97, pp. 6105--6114.

\bibitem{res8}
R.~Tang and J.~Lin,
\newblock ``{Deep Residual Learning for Small-Footprint Keyword Spotting},''
\newblock in {\em Proc. ICASSP}, 2018, pp. 5484--5488.

\bibitem{batchnorm}
S.~Ioffe and C.~Szegedy,
\newblock ``{Batch Normalization: Accelerating Deep Network Training by
  Reducing Internal Covariate Shift},''
\newblock in {\em Proc. ICML}, 2015, vol.~37, pp. 448--456.

\bibitem{accurate_mfcc}
N.~Ryant, M.~Slaney, M.~Liberman, E.~Shriberg, and J.~Yuan,
\newblock ``Highly accurate mandarin tone classification in the absence of
  pitch information,''
\newblock {\em Proc. International Conference on Speech Prosody}, pp. 673--677,
  2014.

\end{thebibliography}

\end{document}